\documentclass[a4paper,10pt]{article}

\usepackage{amsthm}
\usepackage{amsmath}
\usepackage{amssymb}
\usepackage{latexsym}

\newtheorem{theorem}{Theorem}[section]
\newtheorem{proposition}[theorem]{Proposition}
\newtheorem{lemma}[theorem]{Lemma}
\newtheorem{corollary}[theorem]{Corollary}
\newtheorem{remark}[theorem]{Remark}
\newtheorem{definition}[theorem]{Definition}
\newtheorem{example}[theorem]{Example}
\newtheorem{assumption}[theorem]{Assumption}

\begin{document}

\title{Maximizing expected utility in the\\ Arbitrage Pricing Model}

\author{Mikl\'os R\'asonyi\thanks{Alfr\'ed R\'enyi Institute of
Mathematics, Hungarian Academy of Sciences, Budapest.}}

\date{\today}

\maketitle

\begin{abstract}
We consider an infinite dimensional optimization problem motivated by mathematical economics.
Within the celebrated ``Arbitrage Pricing Model'', we use probabilistic and functional analytic techniques
to show the existence of optimal strategies for investors who maximize their expected utility.
\end{abstract}


\section{Introduction}

The present paper treats an infinite dimensional optimization problem coming from microeconomics. 
After sketching the economic theory context, we review the ensuing
developments in mathematical finance and explain certain related challenges in order to put our contribution in context. 
We then provide an overview of the paper.
 
A prevailing tenet of economic theory is that the expected return on a given asset in a financial market
should be an (approximately) linear function of its covariance with the market portfolio,
see e.g. \cite{huang}. In his groundbreaking paper \cite{ross}, S. A. Ross provided a new justification for this
principle, based on arbitrage considerations in a market model with a countably infinite number of assets.
The Arbitrage Pricing Model (APM for short) of \cite{ross} has been put on firm mathematical foundations
in \cite{huberman} and it has since become textbook material for courses on mathematical economics.

Inspired by APM, an abstract framework for markets with ``many'' assets 
has been proposed in \cite{kabanov-kramkov}. The theory of such ``large financial markets'' 
has been extensively investigated in 
\cite{klein-schachermayer1,klein-schachermayer2,kabanov-kramkov1,klein,klein-continuous,
dedonno,largo,def,baran,utility,orlicz,balbas,rokhlin,campi,dinunno}. There 
has also been a recent revival of interest in such models, see \cite{manu,josef}. In the context of
optimization, \cite{paolo} and \cite{oleksi} are the only previous studies we are aware of.

In simple terms, large financial markets are understood to comprise, for all $n\in\mathbb{N}$, a 
market model with $n$ assets and these finite ``market segments'' or ``small markets'' 
are nested into each other.
One may naturally define portfolio strategies on the market segments. However, the set of 
random variables corresponding to portfolio
values thus obtained fails to be closed in any reasonable topology. This is a handicap when studying problems
of optimization where a minimal requirement is that the domain of the objective function should be closed.
For such problems to make sense, a certain closure of the portfolio values needs to
be introduced in order to get a closed domain of optimization. It is desirable that  
every element of this domain has a natural interpretation
in terms of a portfolio involving (possibly) an infinite number of assets, see Remark \ref{palp} below
for details and for comments on previous work.

In the present paper we consider agents whose preferences are of the von Neumann-Morgenstern 
type (see \cite{von} or Chapter 2 of \cite{fs}). We assume that agents have an increasing, 
concave utility function and they aim at maximizing their expected utility from investment returns. 
For this optimization problem we will prove the existence of
an optimal portfolio, see Theorem \ref{bumbo} below. We exhibit a \emph{natural} class
of strategies which is closed (see Lemma \ref{closed}) and in which
an optimizer can be found.

The optimization problem we consider differs from the ones usually studied: it is infinite dimensional not only because the probability space is infinite
(as it is the case in most models of mathematical finance) but also the number
of assets (and hence the dimension of the portfolio) is infinite.
Our arguments rely on probabilistic and functional analytic techniques and exploit
previous results of \cite{largo,def} (which were also based on functional analytic arguments).
In the proof of Theorem \ref{bumbo} we introduce a new recursive procedure exploiting
the duality between portfolio values and equivalent risk-neutral measures. 

While the number of assets is infinite, in the APM there is only one time step. 
It would be desirable to extend our results to the setting of more general large financial markets with 
continuous trading. This is indeed
challenging because finding an appropriate class of admissible strategies seems difficult, see Remark \ref{palp}.

We describe the APM in Section \ref{modell}. In Section \ref{arbitrazs} we review various
arbitrage concepts and see how these are reflected by the model parameters. Section 
\ref{hasznos} presents our main result, Theorem \ref{bumbo}, on the existence of optimal strategies. 
In Corollary \ref{horhos} we also establish the existence of certain equivalent risk-neutral
measures (state price densities) for the APM using utility considerations. Finally, Section \ref{no}
investigates the problem of finding almost optimal strategies which invest into finitely many assets only.

\section{The model}\label{modell}


We fix a probability space $(\Omega,\mathcal{F},P)$. The expectation of a random variable $X$ with respect
to some measure $Q$ on $(\Omega,\mathcal{F})$ is denoted by $E_QX$. If $Q=P$ we drop the
subscript. $L^p(Q)$ is the set of
$p$-integrable random variables with respect to $Q$, for $p\geq 1$. 
If $Q=P$ we simply write $L^p$ instead of $L^p(Q)$ and use $L^{\infty}$ to denote the family of
of (essentially) bounded random variables (with respect to $P$). When $X\in L^2$
we define its variance as
$$
\mathrm{var}(X):=E(X-EX)^2.
$$
The symbol $\sim$ denotes
equivalence of measures, the notation $\mathcal{N}(a,b)$ refers to the Gaussian law
with mean $a$ and variance $b$. 

Following \cite{ross}, our model of a financial market consists of
a countably infinite number of assets/investments. The return on asset $i$ is
the random variable $R_i$, where 
\begin{eqnarray*}
R_0 &:=& r;\quad R_i:=\mu_i+\bar{\beta}_i\varepsilon_i,\quad 1\leq i\leq m;\\
R_i &:=& \mu_i+\sum_{j=1}^m \beta^j_i\varepsilon_j+\bar{\beta}_i
\varepsilon_i,\quad i>m.
\end{eqnarray*}

Asset $0$ is the riskless asset with a constant rate of return $r\in\mathbb{R}$. The random
variables $\varepsilon_i$, $i=1,\ldots,m$ serve
as \emph{factors} which influence the return on all the assets $i\geq 1$ while $\varepsilon_i,\ {i>m}$
are random sources particular to the individual assets $R_i$, $i>m$, which are 
responsible for the so-called ``idiosyncratic risk''. 

For simplicity we set $r=0$ from now on.
We assume that the $\varepsilon_i$ are square-integrable, independent random variables satisfying 
\begin{eqnarray*}
E\varepsilon_i=0,\quad E\varepsilon_i^2=1,\quad i\geq 1,
\end{eqnarray*}
thus the constant $\mu_i$ equals the expected return $ER_i$; the
$\beta^j_i$ and the $\bar{\beta}_i$ are real constants, we assume 
$\bar{\beta}_i\neq 0$, $i\geq 1$.

\begin{remark}
{\rm In \cite{ross,huberman} the $\varepsilon_i$ were assumed merely uncorrelated. We need the stronger
 independence assumption for our purposes. Also, as in \cite{kabanov-kramkov1}, we suppose that
 the ``economic indices'' $R_i$, $i=1,\ldots,m$ are traded assets in our model.
 
 It would be possible to weaken the latter hypothesis: it is enough to assume 
 $R_i=\sum_{j=1}^m a_{ij}(\mu_i+\bar{\beta}_i\varepsilon_i)$ with an invertible matrix 
 $[a_{ij}]_{i,j=1,\ldots,m}$, i.e. the first $m$ assets can be an affine
 transform of the economic factors $\varepsilon_i$, $i=1,\ldots m$. Such a model
 leads to the same family of portfolio values as the model we treat.
 
 The independence of the $\varepsilon_i$, $1\leq i\leq m$ could also be weakened at the price
 of more complicated arguments but we do not see any way to relax the independence
 assumption on the idiosyncratic risk components $\varepsilon_i$, $i>m$.
 We refrain from chasing greater generality here.}
\end{remark}

For any $k\geq 1$ we will call the {\em $\mathrm{k}$th market segment} 
the market model containing the assets
$R_0,\ldots, R_k$ only.
 
\begin{definition}\label{alb} A {\em portfolio} 
$\psi$ in the $\mathrm{k}$th market segment
is an arbitrary sequence $\psi_i,\ 0\leq i\leq k$ 
of real numbers satisfying 
\begin{equation}\label{lb}
\sum_{i=0}^k \psi_i=0.
\end{equation}
\end{definition}

\begin{remark} {\rm This definition can be justified as follows: for each $i\geq 1$ 
we may imagine the existence of an asset
which is worth $1$ dollar at time $0$ (today) and $1+R_i$ dollars at
a fixed future date (tomorrow). Then Definition \ref{alb} is just
the self-financing property of a portfolio in assets $i=0,\ldots,k$ with $0$ initial
capital, see e.g. Chapter 1 of \cite{fs}. One could easily incorporate a nonzero initial
capital as well.}
\end{remark}

For a portfolio $\psi$ in the $k$th market segment, its return is
defined as
$$
V({\psi}):=\sum_{i=0}^{k} \psi_i R_i.
$$
Since $R_0=0$, the constribution of $\psi_0$ does not matter so
\begin{equation}\label{janssons}
V({\psi})=\sum_{i=1}^{k} \psi_i R_i. 
\end{equation}

This means that a portfolio in the $k$th segment is characterized by an
arbitrary sequence $\psi_1,\ldots,\psi_k$ of real numbers which
determine $\psi_0$ (resp. $V(\psi)$) through \eqref{lb} (resp. \eqref{janssons}).

It is convenient to introduce the new parameters 
\begin{eqnarray*}
b_i &:=& -\frac{\mu_i}{\bar{\beta}_i},\quad 1\leq i\leq m;\\
b_i &:=& -\frac{\mu_i}{\bar{\beta}_i}+
\sum_{j=1}^m \frac{\mu_j\beta^j_i}{\bar{\beta}_j\bar{\beta_i}},\quad i>m.
\end{eqnarray*}
Asset returns take the following form:
\begin{eqnarray*}
R_i &=& \bar{\beta}_i(\varepsilon_i-b_i),\quad 1\leq i\leq m;\\
R_i &=& \sum_{j=1}^m\beta_i^j(\varepsilon_j-b_j)+\bar{\beta}_i(\varepsilon_i-
b_i),\quad i>m.
\end{eqnarray*}

It is clear that for any strategy $\psi$ in the $k$th market segment
$$
V(\psi)=\sum_{i=1}^k \phi_i (\varepsilon_i-b_i),
$$
for some real numbers $\phi_i$, $i=1,\ldots,k$ and there is a one-to-one correspondance
between $\phi_i$, $i=1,\ldots,k$ and
$\psi_i$, $i=1,\ldots,k$ (due to $\bar{\beta}_i\neq 0$). 

Let $\mathcal{E}$ denote the set of sequences $\phi=\{\phi_i,\ i\geq 1\}$ of real numbers 
such that, for some $k\geq 1$, $\phi_i=0$, $i\geq k$.
For each $\phi\in\mathcal{E}$ we set
$$
V(\phi):=\sum_{i=1}^{\infty} \phi_i (\varepsilon_i-b_i),
$$
where the sum is, in fact, finite. 

By what has been written above, this is an innocuos abuse of notation since 
the set of random variables which are values of portfolios in the $k$th market segment for some
$k$ equals $K:=\{V(\phi):\phi\in\mathcal{E}\}$. From now on we call the elements of $\mathcal{E}$
\emph{elementary strategies}.

As there is no hope for $K$ to be closed in any useful topology, we should extend the class
of strategies to certain portfolios that may contain infinitely many assets. A mathematically
natural choice is to define the set of \emph{admissible strategies} as $\mathcal{A}:=\ell_2$
where 
$$
\ell_2:=\{(\phi_i)_{i\geq 1}:\ \sum_{i=1}^{\infty}\phi_i^2<\infty\},
$$
that is, the family of square summable sequences which carries a Hilbert space structure
with the norm
$$
\Vert \phi\Vert_{\ell_2}:=\sqrt{\sum_{i=1}^{\infty}\phi_i^2},\quad \phi\in\ell_2.
$$

Assume $\sum_{i=1}^{\infty} b_i^2<\infty$ (we will see in Section \ref{arbitrazs} that
this parameter restriction necessarily holds in an arbitrage-free market).
The $\varepsilon_i$ are orthonormal in $L^2$ hence it is clear that, for any $\phi\in\mathcal{A}$, 
the series $\sum_{i=1}^{\infty}\phi_i(\varepsilon_i-b_i)$ converges in $L^2$ so we may define
$$
V(\phi):=\sum_{i=1}^{\infty}\phi_i(\varepsilon_i-b_i)
$$
as the value of an admissible portfolio $\phi\in\mathcal{A}$. (The series also converges almost
surely by Kolmogorov's theorem, see e.g. Chapter 4 of \cite{kallenberg}.)

\begin{remark} {\rm The economic interpretation of $\mathcal{A}$ is \emph{not} straightforward.
Let $\phi\in\mathcal{A}$ be fixed. 
Taking positions $\phi_1,\ldots,\phi_k$ in assets $1,\ldots,k$ (and zero position in the
remaining assets) forces us to take a position
$-\sum_{i=1}^k \phi_i$ in asset $0$ (in the riskless asset). Now, tending to $\infty$ with 
$k$, it may well happen
that e.g. $-\sum_{i=1}^k \phi_i$ tends to $-\infty$ (we can only guarantee that 
$\sum_{i=1}^k \phi_i^2$ stays bounded). This means that we allow certain portfolios that incur
``infinite debts''. It has to be stressed, however, that accumulating these debts still have certain constraints: namely, $\sum_{i=1}^{\infty} \phi_i^2<\infty$ must hold. Lemma 
\ref{closed} below shows that, economic considerations notwithstanding, the right choice
for the class of admissible strategies is $\mathcal{A}$, as intuition suggested.}
\end{remark}

\begin{remark} {\rm In continuous-time models of mathematical finance the set of admissible
strategies is often restricted, only those are admitted which have a value function bounded
from below (in our context, this would mean $V(\phi)\geq -d$ a.s. for some $d$). This is necessary
since using the full set of strategies would yield arbitrage in those models. 

In the present paper such a restriction is unnecessary (absence of arbitrage
holds under suitable conditions, see Theorem \ref{arbi} below) and it would lead to a trivial
domain of optimization (under Assumption \ref{relevant}, only $\phi=0$ satisfies $V(\phi)\geq -d$ a.s.,
as easily seen). Moreover, even in the case of continuous-time markets, optimizers of utility
maximization problems {\em do not} have a bounded from below value function and the
class of admissible strategies needs to be (carefully) enlarged so as to contain the optimizer,
see e.g. \cite{w}.}
\end{remark}

\section{Concepts of arbitrage}\label{arbitrazs}

The concept of arbitrage is fundamental in mathematical finance. Absence of arbitrage
normally goes together with the existence of risk-neutral measures (elements of the set $\mathcal{M}$
in our setting, see below). Risk-neutral measures not only provide fair pricing rules for
derivative products but they are also important technical tools in optimal investment problems.
See \cite{fs} for more on the interplay of arbitrage, risk-neutral measures and optimal
investment. In large financial markets various concepts of asymptotic arbitrage have
been introduced over the time. Here 
we first recall the original definition from \cite{ross}.

\begin{definition}
We say that {\em there is asymptotic arbitrage} if
there exists a sequence of portfolios $\psi(k)$
in the $\mathrm{k}$th market segment such that the corresponding portfolio returns
$V(\psi(k))$ satisfy
\begin{equation}\label{asym}
EV(\psi(k))\to\infty,\quad
\mathrm{var}(V({\psi(k)}))\to 0,
\end{equation}
as ${k}\to\infty$.
If there exist no such sequence we say that 
{\em there is no asymptotic arbitrage}.
\end{definition}

\begin{definition}\label{nafl}
We say that {\em there is an asymptotic free lunch} 
if there exists a sequence of trading strategies
$\phi(k)\in\mathcal{E}$ such that
$$
V(\phi(k))\to X,\ k\to\infty
$$
in probability where $X$ is a $[0,\infty]$-valued random variable
with $P(X>0)>0$. If no such sequence exists then we say that
{\em there is no asymptotic free lunch}.
\end{definition}

Absence of arbitrage admits diverse formulations and each of these is usually
equivalent to the existence of some dual variables corresponding to pricing functionals, see \cite{fs} about the classical
case of finitely many assets
and \cite{kabanov-kramkov,klein-schachermayer1,klein-schachermayer2,kabanov-kramkov1,klein,klein-continuous,
largo,def} about large financial markets. The next definition introduces the dual variables of the present paper.

\begin{definition}
We call a probability $Q\sim P$ an {\em equivalent risk-neutral measure}
for the given market 
if, for all $i\geq 1$, 
\begin{equation}\label{nulla}
E_Q R_i=0,
\end{equation}
that is, if the expected return under $Q$ on all assets equals the riskless return.
The set of equivalent risk-neutral measures is denoted by $\mathcal{M}$.
\end{definition}

It is clear that $Q\in\mathcal{M}$ iff $E_Q \varepsilon_i=b_i$ for all $i\geq 1$.
We need that \eqref{nulla} extends to all admissible portfolios for suitable $Q$ 
and $b_i$.

\begin{lemma}\label{varhato}
 Let $Q\in\mathcal{M}$ with $dQ/dP\in L^2$ and assume  $\sum_{i=1}^{\infty}b_i^2<\infty$.
 Then, for all $\phi\in\mathcal{A}$,
 $$
 E_QV(\phi)=0.
 $$
\end{lemma}
\begin{proof}
 It is enough to show that the series $\sum_{i=1}^{\infty} \phi_i(\varepsilon_i-b_i)$
 converges in $L^1(Q)$. For any $n,m$, the Cauchy inequality implies
 \begin{eqnarray*}
 E_Q\left|\sum_{i=n}^m \phi_i(\varepsilon_i-b_i)\right|\leq \sqrt{E(dQ/dP)^2}\sqrt{E
 \left(\sum_{i=n}^m \phi_i(\varepsilon_i-b_i)\right)^2}
 \end{eqnarray*}
 and $\sum_{i=1}^{\infty} \phi_i(\varepsilon_i-b_i)$ is convergent in $L^2$, which proves
 $L^1(Q)$-convergence.
\end{proof}

Our main result needs stronger hypotheses on the $\varepsilon_i$ which we now present.

\begin{assumption}\label{relevant}
For each $x\geq 0$, both 
\begin{equation}\label{vavelgrof}
\inf_{i\geq 1}P(\varepsilon_i>x)>0\mbox{ and }
\inf_{i\geq 1}P(\varepsilon_i<-x)>0 
\end{equation}
hold. Furthermore, 
\begin{equation}\label{unni}
\sup_{i\in\mathbb{N}}E[\varepsilon_i^2 1_{\{|\varepsilon_i|\geq N\}}]\to 0,\ N\to\infty.
\end{equation}
\end{assumption}

\begin{remark}
{\rm Notice that \eqref{unni} is just the uniform
integrablity of the family $\varepsilon_i^2$,
$i\in\mathbb{N}$ of random variables. This is a mild assumption on the moments of the $\varepsilon_i$. 

The hypotheses \eqref{vavelgrof} appear to be more stringent. We remark, however, that assuming the presence
of extreme risks (i.e. $\varepsilon_i$ having unbounded support on both the positive and
negative axes) is a standard feature of most market models. Consider, e.g., the Black-Scholes
market (see \cite{bjork}) where the prices $S_{T_1},S_{T_2}$ of a given stock
at times $0<T_1<T_2$ are
such that $S_{T_2}=S_{T_1}X(T_1,T_2)$ with $S_{T_1}$, $X(T_1,T_2)$ independent
lognormal random variables. In this case the return on the stock, $S_{T_2}-S_{T_1}$ has 
a law with unbounded support in both directions.

Roughly speaking, the conditions in \eqref{vavelgrof}
stipulate that the risks of a (possibly extreme) price change are ``uniformly present'' for each asset.

For instance, if $\mathcal{P}:=\{\mathrm{Law}(\varepsilon_i):\ i\geq 1\}$ is a finite set and
each of its elements has a a support which is unbounded on $\mathbb{R}$ in both directions then 
Assumption \ref{relevant} holds. It is also very easy to provide suitable infinite families
$\mathcal{P}$: let each $\mu\in\mathcal{P}$ satisfy 
\begin{eqnarray*}
h(z) &\leq& \mu(\{x\leq -z\})\leq Cz^{-\theta}\\
h(z) &\leq& \mu(\{x\geq z\})\leq Cz^{-\theta}
\end{eqnarray*}
for all $z\geq 1$ with constants $C>0$, $\theta>2$ and with some function $h:[1,\infty)\to (0,\infty)$ (e.g.
$h(z)=c z^{-\eta}$ with some $\eta\geq \theta$ and with another constant $c>0$ or
$h(z)=c_1e^{-c_2 z^2}$ with $c_1,c_2>0$).}
\end{remark}

The following theorem clarifies the relationship between market parameters and various
concepts of absence of arbitrage.
\begin{theorem}\label{arbi}
Consider the following conditions.
\begin{enumerate}
 \item There is no asymptotic free lunch.
 
 \item For each $P'\sim P$ there exists $Q\in\mathcal{M}$ with $dQ/dP'\in L^{\infty}$.
 
 \item There exists $Q\in\mathcal{M}$ with $dQ/dP\in L^2$.
 
 \item $$\sum_{i=1}^{\infty} b_i^2<\infty$$

 \item There is no asymptotic arbitrage.
\end{enumerate}
Then $1.\Leftrightarrow 2.\Rightarrow 3. \Rightarrow 4.\Leftrightarrow 5.$ holds. Under Assumption \ref{relevant},
all the conditions $1.$-$5.$ are equivalent.
\end{theorem}
\begin{proof}
The equivalence $1.\Leftrightarrow 2.$ follows from Theorem 1 of \cite{largo} (and it is true in
much greater generality, by infinite dimensional separation arguments for convex sets). 
$4.\Leftrightarrow 5.$, $3.\Rightarrow 4.$ are just Theorem 1  
and Proposition 3 of \cite{def}, respectively (deduced again by functional analytic arguments). 
Since $2.\Rightarrow 3.$ is trivial, it suffices to establish $4.\Rightarrow 1.$ under Assumption
\ref{relevant}. 

Let $\phi(n)\in\mathcal{E}$ such that $V(\phi(n))\to X$ a.s. for some random variable $X\in [0,\infty]$.
Let $k_n$ be such that $\phi_l(n)=0$ for
$l\geq k_n$. We may and will assume that $k_{n+1}>k_n$ for all $n$.
First let us consider the case where
$\sup_n ||\phi(n)||_{\ell_2}=\infty$. By extracting a subsequence (which we continue to denote by $n$)
we may and will assume $||\phi(n)||_{\ell_2}\to\infty$, $n\to\infty$.
Define $\tilde{\phi}_i(n):=
\phi_i(n)/||\phi(n)||_{\ell_2}$ for all $n,i$. Clearly, $\tilde{\phi}(n)\in\mathcal{E}$
and
\begin{equation}\label{harom}
\liminf_{n\to\infty}V(\tilde{\phi}(n))=\liminf_{n\to\infty}
\frac{V(\phi(n))}{\Vert\phi(n)\Vert_{\ell_2}}\geq 0\mbox{ a.s.} 
\end{equation}
Let $M:=\sqrt{\sum_{i=1}^{\infty} b_i^2}$.
We obviously have $|\sum_{i=1}^{\infty} \tilde{\phi}_i(n)b_i|\leq M$ for all $n$ since
$\Vert\tilde{\phi}(n)\Vert_{\ell_2}=1$.
Hence, along a subsequence (still denoted by $n$), one has 
$$
\sum_{i=1}^{\infty}\tilde{\phi}_i(n)b_i\to d,\quad n\to\infty
$$
for some $d\in\mathbb{R}$.
We now distinguish two subcases.

\textsl{Subcase 1:} When $\chi_n:=\max\{|\tilde{\phi}_i(n)|,\ i=1,\ldots,k_n\}\to 0$, $n\to\infty$. 
Fix $\eta,\delta>0$. Assumption \ref{relevant} implies that, for some $N=N(\delta)$, we have
$E[\varepsilon_i^21_{\{|\varepsilon_i|\geq N\}}]<\delta$ for all $i\geq 1$. Choose $n$ so large that $\eta/\chi_n\geq N$.
Notice that $\mathrm{var}(\sum_{i=1}^{k_n} \tilde{\phi}_i(n) \varepsilon_i)=1$ and  
$$
\sum_{i=1}^{k_n} E[\tilde{\phi}_i^2(n)\varepsilon_i^21_{\{|\tilde{\phi}_i(n)\varepsilon_i|\geq \eta\}}] 
\leq \sum_{i=1}^{k_n} E[\tilde{\phi}_i^2(n)\varepsilon_i^21_{\{{\chi}_n|\varepsilon_i|\geq \eta\}}]\leq
\delta \sum_{i=1}^{k_n} \tilde{\phi}_i^2(n)=\delta.
$$
As this works for arbitrary $\eta,\delta$, we conclude that the Lindeberg condition holds
for the sums $V(\tilde{\phi}(n))$, $n\geq 1$, so 
the the central limit theorem (see e.g.
Chapter 9 of \cite{chow})
applies. It follows that $\mathrm{Law}(V(\tilde{\phi}(n)))\to \mathcal{N}(-d,1)$ weakly
as $n\to\infty$. In particular, 
$P(V(\tilde{\phi}(n))<0)\to f$ for some $f>0$, an immediate contradiction with \eqref{harom}. So
Subcase 1 actually never occurs.

\textsl{Subcase 2:} When there is $c>0$ and $1\leq l(n)\leq k_n$ such that $|\tilde{\phi}_{l(n)}(n)|\geq c$,
for all $n$. Set $J_i:=\varepsilon_i-b_i$, $i\geq 1$. Note that 
\[
\left\vert\sum_{i\neq l(n)}\tilde{\phi}_i(n)b_i\right\vert\leq 
\Vert\tilde{\phi}(n)\Vert_{\ell_2}\, \Vert b\Vert_{\ell_2}\leq M,
\]
so we have, by Markov's inequality, for each $H>0$,
\begin{eqnarray}\nonumber
P\left(\sum_{i\neq l(n)}\tilde{\phi}_i(n)J_i>H\right) &\leq& 
\frac{\left(\sum_{i\neq l(n)}\tilde{\phi}_i(n)b_i\right)^2+E\left(\sum_{i\neq l(n)}\tilde{\phi}_i(n)\varepsilon_i\right)^2}{H^2}
\leq\\
\frac{M^2 + \sum_{i\neq l(n)}\tilde{\phi}_i(n)^2}{H^2} &\leq&
\frac{M^2 +1}{H^2}\to 0,\label{markov}
\end{eqnarray}
as $H\to\infty$, uniformly in $n$, in particular, 
$P(\sum_{i\neq l(n)}\tilde{\phi}_i(n)J_i\leq H)\geq 1/2$ holds for all $n$ with
some $H> 0$ large enough. However, Assumption \ref{relevant} implies that there is $q>0$ such that
$$
P(\varepsilon_{l(n)}<-(H+M+1)/c)\geq q,\ P(\varepsilon_{l(n)}>(H+M+1)/c)\geq q 
$$
both hold for all $n$. Note that $|\tilde{\phi}_{l(n)}(n)b_{l(n)}|\leq M$. Thus we get
$$
P(V(\tilde{\phi}(n))\leq -1)\geq P\left(\tilde{\phi}_{l(n)}(n)J_{l(n)}\leq -H-1,
\sum_{i\neq l(n)}\tilde{\phi}_i(n)J_i\leq H\right)\geq q/2,
$$
by independence of $\varepsilon_{l(n)}$ from $\varepsilon_i$, $i\neq l(n)$, another
contradiction with \eqref{harom}. This means that Subcase 2 never occurs either.

We now turn to the case $\sup_n ||\phi(n)||_{\ell_2}<\infty$. Then there is a 
subsequence which is weakly convergent in $\ell_2$ and, by the Banach-Saks theorem applied in $\ell_2$, 
convex combinations of $\phi(n)$
(which we continue to denote $\phi(n)$) satisfy
$$
\Vert\phi(n)-\phi^*\Vert_{\ell_2}=\sum_{i=1}^{\infty} (\phi_i(n)-\phi_i^*)^2\to 0,\ n\to\infty
$$
for some $\phi^*\in\mathcal{A}=\ell_2$. Hence, by orthonormality of the system $\varepsilon_i$, $i\geq 1$
in $L^2$, 
$$
E(V(\phi(n))-V(\phi^*))^2\to 0,\quad n\to\infty.
$$
Since $L^2$ convergence implies convergence in probability, we get
$V(\phi^*)=X$. If $\phi^*_i=0$ for all $i$ then $X=0$ and we are done. Otherwise
there would exist $l\geq 1$ with, say, $\phi^*_l>0$ (the case $\phi^*_l<0$ follows similarly).
We will show that this cannot happen.

Indeed, as in \eqref{markov},
$$
P\left(\sum_{i\neq l}{\phi}_i^* J_i>H\right)\leq \frac{M^2+\sum_{i\neq l}
({\phi}_i^*)^2}{H^2}\leq \frac{M^2+1}{H^2}
<1/2
$$
when $H$ is large enough. Then
$$
P(V({\phi}^*)\leq -1)\geq P\left({\phi}^*_{l}J_{l}\leq -H-1,
\sum_{i\neq l}{\phi}^*_i(n)J_i\leq H\right)>0,
$$
by independence and by Assumption \ref{relevant}, a contradiction with the fact that $V(\phi^*)=X\geq 0$.
\end{proof}

\begin{remark}
{\rm An important message of Theorem \ref{arbi} is that if APM is arbitrage free then
condition $4.$ should hold. Using terms of mathematical economics, this latter
condition means that the total squared Sharpe ratio of the given market is finite.

For us, however, the usefulness of Theorem \ref{arbi} comes from $2.$: it provides
$Q\in\mathcal{M}$ with a strong additional property (bounded $P$-density).} 
\end{remark}

\begin{example}\label{aba}
{\rm We do not know whether Assumption \ref{relevant} can be weakened. However, some additional assumptions
on the $\varepsilon_n$ are necessary for the validity of Theorem \ref{arbi}, 
as the following example shows.

Let $b_n=0$ for all $n\geq 1$ and let $\varepsilon_n$, $n\geq 2$ be independent with law
$$
P\left(\varepsilon_n=
\frac{\frac{1}{n}+\frac{1}{n^3}}{\sqrt{1+\frac{1}{n^2}-\frac{1}{n^4}-\frac{1}{n^6}}}\right)=1-\frac{1}{n^2},\ 
P\left(\varepsilon_n=
\frac{-n+\frac{1}{n^3}}{\sqrt{1+\frac{1}{n^2}-\frac{1}{n^4}-\frac{1}{n^6}}}\right)=\frac{1}{n^2}, 
$$
for all $n\geq 2$. One can check that $E\varepsilon_n=0$ and $E\varepsilon_n^2=1$, for all $n\geq 2$.
Let $\varepsilon_1$ be arbitrary, independent of the $\varepsilon_n$, $n\geq 2$ with zero mean and
unit variance.

Defining $\phi_i(k):=1$, $2\leq i\leq k$, $\phi_i(k):=0$, $i>k$, $i=1$, the Borel-Cantelli lemma
shows that for a.e. $\omega\in\Omega$ there is $m=m(\omega)$ such that $\varepsilon_i(\omega)\geq 1/(2i)$
for $i\geq m$ which means that
$$
V(\phi(k))(\omega)=\sum_{i=2}^k \varepsilon_i(\omega)\to\infty
$$
a.s. when $k\to\infty$, i.e. this sequence of strategies produces an asymptotic free lunch, even though
$b_n=0$ holds for all $n$. 

It follows that the implication $4.\Rightarrow 1.$ (on which our arguments in Section
\ref{hasznos} hinge) fails in the example just presented. 
}
\end{example}

Denote by $\overline{K}$ the closure of $K$
for convergence in probability. We now exploit the argument of Theorem \ref{arbi} to
prove the following natural, but far from obvious result. Define $K_1:=\{V(\phi):\phi\in\mathcal{A}\}$.
\begin{lemma}\label{closed} 
Under Assumption \ref{relevant} and $\sum_{i=1}^{\infty}b_i^2<\infty$,
\begin{equation}\label{rhs}
\overline{K}=K_1. 
\end{equation}
\end{lemma}
\begin{proof}
 Note that $\mathcal{E}\subset\mathcal{A}$ and for every $\phi\in\mathcal{A}$,
 $V(\phi(n))\to V(\phi)$ in $L^2$ and hence also in probability, where $\phi(n)\in\mathcal{E}$ is
 such that $\phi_i(n)=\phi_i$,
 $i\leq n$, $\phi_i(n)=0$, $i>n$. Thus it is enough to prove that $K_1$ is
 closed in probability. 
 
 Let $\phi(n)\in\mathcal{A}$ be such that $V(\phi(n))\to X$ almost
 surely with some (finite-valued) random variable $X$ as $n\to\infty$. We may and will assume $\phi(n)\in\mathcal{E}$, $n\in\mathbb{N}$.
 
 If we had $\sup_n||\phi(n)||_{\ell_2}=\infty$ then, repeating the argument of Theorem \ref{arbi} 
 (Subcases 1 and 2), we get that $V(\tilde{\phi}(n))\to 0$ a.s. and 
 $$
 \liminf_{n\to\infty}P(V(\tilde{\phi}(n))<0)\geq f\mbox{ for some }f>0,
 $$
 both hold: a contradiction. Hence $\sup_n||\phi(n)||_{\ell_2}<\infty$ and we can find,
 just like in the proof of Theorem \ref{arbi}, $\phi^*\in\mathcal{A}$ with $V(\phi^*)=X$.
\end{proof}

\begin{example}{\rm By the orthonormality of $\varepsilon_i$, $i\geq 1$ in $L^2$ it is obvious that
$K_1$ is closed in $L^2$. This, however, does not suffice in
subsequent arguments for proving the existence of optimal strategies, we need
the stronger property of being closed in probability.

This latter property can easily fail: consider the model described in Example \ref{aba} and
the sequence of strategies $\lambda_i(k):=1/\ln(k)$, $2\leq i\leq k$, $\lambda_i(k):=0$, $i>k$, $i=1$. By the Borel-Cantelli lemma, for a.e. $\omega$, only finitely
many terms of the series $\sum_{i=1}^{k}
\varepsilon_i(\omega)$, $k\geq 1$ differ from those of $$
\sum_{i=1}^k \frac{\frac{1}{i}+\frac{1}{i^3}}{\sqrt{1+\frac{1}{i^2}-\frac{1}{i^4}-\frac{1}{i^6}}}
$$
and the latter is asymptotically of the order $\ln(k)$. It follows that  
$V(\lambda(k))\to 1$ a.s. as $k\to\infty$ and $X(\omega):=1$, $\omega\in\Omega$ is not 
in $K_1$
since it is orthogonal to all of the $\varepsilon_i$ in the Hilbert space $L^2$. It is also easily
seen that, with the same sequence $\varepsilon_i$, $i\geq 1$, 
closedness of $K_1$ fails even for an arbitrary sequence $b_i$, $i\geq 1$ satisfying
$\sum_{i=1}^{\infty} |b_i|<\infty$: indeed, $X\in \overline{K}$ as before but $\mathrm{var}(X)=0$
while $\mathrm{var}(V(\phi))\neq 0$ for each $\phi\in\mathcal{A}$ which is not identically $0$, hence
$K_1$ is \emph{not} closed in probability.

It would be nice to know how Assumption \ref{relevant} can be weakened  
in Theorem \ref{bumbo}. It seems highly unlikely to prove the existence of an optimizer in the setting of
Section \ref{hasznos} without $K_1$ being closed in probability.
Consequently, the counterexample just sketched suggests that some additional assumptions on the 
$\varepsilon_i$
are needed (i.e. independence, mean zero and unit variance alone do not 
suffice).}
\end{example}

\section{Utility maximization}\label{hasznos}

We consider an investor with utility function $u:\mathbb{R}\to\mathbb{R}$. We assume
in the rest of the paper that $u$ is concave
and non-decreasing. Concavity expresses risk aversion, the non-decreasing property
means that the investor prefers more money to less.
We will use the following simple lemma.

\begin{lemma}\label{lena}
 If $u$ is not constant then there exist $c,C>0$ such that $u(x)\leq -c|x|+C$ for
 all $x\leq 0$. 
\end{lemma}
\begin{proof}
As we have $u(-\infty)=-\infty$ in this case, there is $x^*\leq 0$ such that 
$u(x)$, $x\leq x^*+1$ is a strictly increasing function and $u(x^*)<0$. Denoting by $d^*:=u'(x^*-)>0$
its left-hand side derivative, we have
$$
u(x)\leq u(x^*)+(x-x^*)d^*,\ x\leq x^*.
$$
On the other hand, $$
u(x)\leq |u(0)|\leq -d^*|x|+d^*|x^*|+|u(0)|
$$ 
for $x^*\leq x\leq 0$.
Setting $c:=d^*$ and $C:=d^*|x^*|+|u(0)|$, the statement is shown.
\end{proof}

For $x\in\mathbb{R}$, we denote $x^+:=\max\{0,x\}$, $x^-:=\max\{-x,0\}$.

\begin{remark}\label{cista}
{\rm When $u(0)=0$, it follows from Lemma \ref{lena} that, for any $x\leq 0$,
\[
|x|\leq \frac{u^-(x)}{c}+\frac{C}{c}.
\]
}
\end{remark}

We first assert the existence of an optimal investment in the case where the investor's utility function is
bounded from above. The proof will be based on Theorem \ref{arbi}.

\begin{proposition}\label{main} Assume the existence of $Q\in\mathcal{M}$ with $dQ/dP\in L^{\infty}$.
Let $u:\mathbb{R}\to\mathbb{R}$ be bounded from above. Then there exists
$X\in\overline{K}_1$ such that 
\begin{equation}\label{force1}
Eu(X)=\sup_{\phi\in\mathcal{A}}Eu(V(\phi)),
\end{equation}
where $\overline{K}_1$ denotes the closure of $K_1$ with respect to convergence in probability.
\end{proposition}
\begin{proof} Note that, $u$ being bounded above, $Eu(V(\phi))$ makes sense for all $\phi\in\mathcal{A}$. 

Let $\phi(n)$ be a sequence such that $\sup_{\phi\in\mathcal{A}}Eu(V(\phi))=\lim_{n\to\infty}
Eu(V(\phi(n)))$ and $Eu(V(\phi(n)))>-\infty$ for all $n$. Then, since $u$ is bounded from above, $\sup_n Eu^-(V(\phi(n)))<\infty$.
The case of constant $u$ is trivial, otherwise 
$$
\sup_n EV^-(\phi(n))<\infty
$$ 
holds, by Lemma \ref{lena}. Then also $\sup_n E_QV^-(\phi(n))<\infty$.
It follows from Lemma \ref{varhato} that $E_QV(\phi(n))=0$ for all $n$,
consequently, $\sup_n E_Q|V(\phi(n))|<\infty$.

By the Koml\'os theorem (see \cite{komlos}) applied in $L^1(Q)$, convex combinations of $V(\phi(n))$ converge to some
random variable $X$ almost surely. Since ${K}_1$ is trivially convex, $X\in\overline{K}_1$. By concavity of $u$
and the (reverse) Fatou lemma, $Eu(X)\geq \lim_{n\to\infty}Eu(V(\phi(n)))$.
\end{proof}

\begin{corollary}\label{akutt} Let Assumption \ref{relevant} be in force and assume $\sum_{i=1}^{\infty}b_i^2<\infty$.
Let $u:\mathbb{R}\to\mathbb{R}$ be bounded from above. Then there exists
$\phi^*\in\mathcal{A}$ such that 
\begin{equation}\label{force}
Eu(V(\phi^*))=\sup_{\phi\in\mathcal{A}}Eu(V(\phi)).
\end{equation}
\end{corollary}
\begin{proof} The hypotheses imply the existence of $Q\in\mathcal{M}$ with $dQ/dP\in L^{\infty}$, by Theorem \ref{arbi}. The set
${K}_1$ is closed by Lemma
\ref{closed}, hence there is $\phi^*\in\mathcal{A}$ such that $X=V(\phi^*)$. 
\end{proof}

\begin{corollary}\label{mobutu}
Let Assumption \ref{relevant} be in force, assume $\sum_{i=1}^{\infty}b_i^2<\infty$
and that $Eu(V(\phi))$ is finite for 
all $\phi\in\mathcal{A}$. Furthermore, let $u$ be strictly increasing and continuously
differentiable with bounded $u'$. If \eqref{force} holds then there exists $Q\in\mathcal{M}$ such that
$$
\frac{dQ}{dP}=\frac{u'\left(\sum_{i=1}^{\infty} \phi^*_i (\varepsilon_i-b_i)\right)}
{Eu'\left(\sum_{i=1}^{\infty} \phi^*_i (\varepsilon_i-b_i)\right)}.
$$
\end{corollary}
\begin{proof} Fix $l\in\mathbb{N}$ and let $\phi^*$ be an optimal strategy as in \eqref{force}.
Consider the function $g(x):=Eu(x J_l+\sum_{i\neq l}\phi^*_i J_i)$, $x\in\mathbb{R}$,
where $J_i=\varepsilon_i-b_i$, $i\geq 1$.
Clearly, $g$ attains its maximum at $x=\phi^*_l$. By the mean value theorem, for each 
$\phi\in\mathbb{R}$ and $h\in (-1,1)$
we have $$
\frac{1}{h}\left[u\left([\phi+h]J_l+\sum_{i\neq l}\phi^*_i J_i\right)-u\left(\phi J_l+
\sum_{i\neq l}\phi^*_i J_i\right)\right]
=u'\left( \xi(h)J_l+\sum_{i\neq l}\phi^*_i J_i\right)J_l
$$
for some random variable $\xi(h)$ between $\phi$ and $\phi+h$. 
We let $h\to 0$. Since $u'$ is bounded, Lebesgue's theorem implies
that $g'(\phi)$ exists and equals $E[u'(\phi J_l+\sum_{i\neq l}\phi^*_i J_i)J_l]$.
It follows that
$$
0=g'(\phi^*_l)=E\left[u'\left( \phi^*_lJ_l+\sum_{i\neq l}\phi^*_i J_i\right)J_l\right].
$$
As this holds for each $l$, we get that $Q\in\mathcal{M}$ noting that
$u$ is strictly increasing hence $u'(x)>0$ for all $x$.
\end{proof}

\begin{remark}{\rm
The construction of Corollary \ref{mobutu} is standard in the context of finitely many assets,
see \cite{davis}.} 
\end{remark}

Now we turn to the case of $u$ not necessarily bounded above.
In this setting the domain of optimization will be
$$
\mathcal{A}'(u):=\{\phi\in\mathcal{A}:\ Eu^-(V(\phi))<\infty\}
$$
so that the expectation $Eu(V(\phi))$ makes sense for all $\phi\in\mathcal{A}'(u)$.
Note that $\mathcal{A}'(u)$ is never empty, $\phi_i=0$, $i\geq 1$ is therein.
The next result is the main theorem of the present paper.

\begin{theorem}\label{bumbo}
 Let Assumption \ref{relevant} be in force and assume 
\begin{equation}\label{manyineni}
\sum_{i=1}^{\infty}b_i^2<\infty.
\end{equation} 
Let $u$ be such that for some constant $C_1\geq 0$,
 \begin{equation}\label{retes}
 u(x)\leq C_1(x^{\alpha}+1),\ x\geq 0.
 \end{equation}
with $0\leq \alpha<1$. Then
 there exists
$\phi^*\in\mathcal{A}'(u)$ such that 
\begin{equation}\label{aki}
Eu(V(\phi^*))=\sup_{\phi\in\mathcal{A}'(u)}Eu(V(\phi))<\infty.
\end{equation}
\end{theorem}

\begin{remark}\label{palp} {\rm The only papers in the 
existing literature that are closely related to ours are \cite{paolo}
and \cite{oleksi}, where the expected
utility of an investor is maximized in a continuous-time large financial market
over a set of ``generalized portfolios''. The first paper is about
maximizing terminal utility while the second paper deals with utility
from consumption, allowing random utility functions and a stochastic clock.
Note that, in both of those papers, the utility function $u$ is defined on the positive axis 
$\mathbb{R}_+$ while here we consider $u:\mathbb{R}\to\mathbb{R}$ which allows for an analysis of risks related to losses. See the very recent \cite{leguj} about some
progress in continuous-time large markets with $u:\mathbb{R}\to\mathbb{R}$. 

The ``generalized portfolios'' are processes which
are in the closure of the value processes of portfolios
in the market segments in a suitable topology (the \'Emery topology). 
This means that the value of a generalized portfolio can be
approximated with arbitrary precision by values of investment opportunities in
finitely many assets. 

Such a choice is reasonable and, by examples of \cite{josef}, probably inevitable.
We would, however, like to see generalized portfolios represented
as investments into infinitely many assets, with the size of the position
in each asset explicitly given. Theorem \ref{bumbo} above finds the optimizer in the 
class $\mathcal{A}'(u)$ whose elements have an obvious interpretation in terms of portfolios in infinitely many assets.

To make our point more clear, let $\varepsilon_i$ be an \emph{arbitrary} sequence of random variables, for a moment. 
The closure of $K$ in probability has no intrinsic characterization in that case, i.e.
for some $X\in \overline{K}$ we may find $\phi(n)\in\mathcal{E}$ such that 
$V(\phi(n))\to X$ almost surely, but it is by no means easy to decide whether there 
exists some
sequence $\phi$ with $X=\sum_{i=1}^{\infty} \phi_i (\varepsilon_i-b_i)$. Example \ref{aba} shows that this
problem may arise even in the case of independent $\varepsilon_i$. 
However, in the very specific setting of the Arbitrage Pricing Model and stipulating
Assumption \ref{relevant}, such a characterization of $\overline{K}$ becomes possible.
 
We also refer to the recent, smaller scale companion paper \cite{ijtaf}, where \eqref{vavelgrof} is relaxed to
a simple no-arbitrage condition at the price of requiring uniform \emph{exponential} integrability
of the $\varepsilon_i$ instead of \eqref{unni}.}
\end{remark}

\begin{remark}{\rm
The standard assumption on $u$ in continuous-time semimartingale models is \emph{reasonable
asymptotic elasticity}, see \cite{w}. In discrete-time markets this condition can be
slightly weakened, \cite{rs}. In the present setting even less (namely,
\eqref{retes} above) suffices. We cannot allow $\alpha=1$ in \eqref{retes} as the conclusion
of Theorem \ref{bumbo} clearly fails for linear $u$.}
\end{remark}

\begin{proof}[Proof of Theorem \ref{bumbo}.]
The argument is based on a recursive procedure
using Proposition \ref{main} and Corollary \ref{mobutu} as its springboard.
Assuming that Theorem \ref{bumbo} has been established for some $1>\alpha_n\geq 0$, we use Corollary
\ref{mobutu} to obtain an element of $\mathcal{M}$ which, in turn, helps to establish Theorem \ref{bumbo}
for some $\alpha_{n+1}>\alpha_n$. As $\alpha_n\to 1$, $n\to\infty$ will hold, we will finally get Theorem \ref{bumbo}
for all $\alpha<1$.

Fix $0<\epsilon<1$ and let us define 
$$
u_1(x):=\epsilon x-1,\ x<0,\quad u_1(x):=\frac{-1}{(1+x)^{\epsilon}},\ x\geq 0.
$$
This is a concave and continuously differentiable function to which Corollaries \ref{akutt} and  \ref{mobutu}
clearly apply hence we get $Q\in\mathcal{M}$ with $dQ/dP\in L^{\infty}$ and
$dP/dQ\in L^{2/(1+\epsilon)}$ since $\sum_{i=1}^{\infty}\phi^*_i(\varepsilon_i-b_i)\in L^2$ and
$1/u_1'(x)=(1+x)^{\epsilon+1}/\epsilon$ for $x\geq 0$. Since $2/(1+\epsilon)\to 2$ as $\epsilon\to 0$,
we can conclude that for $p_1:=2$ there is $Q_1=Q_1(\epsilon)\in\mathcal{M}$ with $dQ_1/dP\in L^{\infty}$
and $dP/dQ_1\in L^{p_1-\epsilon}$, for each $\epsilon$.

Let us now suppose that, for $n\geq 1$, $p_n\geq 2$ and for any $0<\epsilon<1$ 
there exists $Q_n=Q_n(\epsilon)\in\mathcal{M}$ with $dQ_n/dP\in L^{\infty}$ and
$dP/dQ_n\in L^{p_n-\epsilon}$. We proceed to show that this implies the existence of
an optimiser for any concave, nondecreasing $u$  satisfying \eqref{retes}
with $\alpha<p_n/(p_n+1)$.

Fix $Q\in\mathcal{M}$. By adding a constant to $u$, we may and will assume $u(0)=0$. 
Adding a constant preserves the validity of \eqref{retes} (perhaps with a different $C_1$)
and does not affect $\mathcal{A}'(u)$ either. The case of constant $u$ is trivial, otherwise
we may and will assume that the conclusion of Lemma \ref{lena} holds.

Let $\phi(k)\in\mathcal{A}'(u)$, $k\in\mathbb{N}$ be such that
 $$
 Eu(V(\phi(k))\to \sup_{\phi\in\mathcal{A}'(u)}Eu(V(\phi)),\ k\to\infty.
 $$

 The following estimations are inspired by Lemma 3.13 of \cite{rr}. 
 Assuming that $dQ/dP$ is bounded above, $E_Q(V(\phi))=0$ follows for each $\phi\in\mathcal{A}$
 by Lemma \ref{varhato}.  
 So for any $\phi\in\mathcal{A}'(u)$ we get, applying H\"older's inequality and Remark \ref{cista},
 \begin{eqnarray}\nonumber
 Eu^+(V(\phi)) &\leq& C_1+C_1 E[V^+(\phi)^{\alpha}]\leq C_1+C_2 (E_QV^+(\phi))^{\alpha}\\
\nonumber =C_1+C_2( E_Q V^-(\phi))^{\alpha} &\leq& C_1+C_3 (EV^-(\phi))^{\alpha}\\  
 &\leq&
\label{moon} C_4(Eu^-(V(\phi)))^{\alpha}+C_4
 \end{eqnarray}
with constants $C_2:=C_1[E(dP/dQ)^{\alpha/(1-\alpha)}]^{1-\alpha}$,
$C_3:=C_2(\mathrm{ess.}\sup dQ/dP)^{\alpha}$, 
$C_4:=\max\{ C_1+C_3C^{\alpha}/c^{\alpha}, C_3/c^{\alpha}\}$, 
where $c,C$ are from Lemma \ref{lena} and we used the fact that $(x+y)^{\alpha}\leq x^{\alpha}+y^{\alpha}$,
$x,y\geq 0$. Notice that $\alpha/(1-\alpha)<p_n$ so
one can choose $Q:=Q_n$ such that $C_2<\infty$ and also $dQ_n/dP\in L^{\infty}$ hence $C_3,C_4<\infty$ hold
as well. We can even assume that $E(dP/dQ)^{\theta/(1-\theta)}<\infty$ for some $\alpha<\theta<1$.

If we had $\sup_k E_Q ([V(\phi(k))]^+)=\infty$ then, by \eqref{moon},
(along a subsequence still denoted by $k$) $Eu^-(V(\phi(k)))\to\infty$, $k\to\infty$ would hold,
hence  
\begin{eqnarray*}
\nonumber Eu(V(\phi(k)))&=& Eu^+(V(\phi(k)))-Eu^-(V(\phi(k)))\\
\nonumber &\leq& C_4(Eu^-(V(\phi(k))))^{\alpha}+C_4-Eu^-(V(\phi(k))),
\end{eqnarray*}
which tends to $-\infty$ as $k\to\infty$, a contradiction. 
It follows that $\sup_k E_Q |V(\phi(k))|=2\sup_k E_Q ([V(\phi(k))]^+)<\infty$ and the Koml\'os
theorem implies the existence of convex combinations of $\phi(k)$ (still denoted by $\phi(k)$)
satsifying $V(\phi(k))\to X$ a.s. for some random variable $X$. By convexity and closedness of $K_1$
(see Lemma \ref{closed}) we get that $X=V(\phi^*)$ for some $\phi^*\in\mathcal{A}$. 

Since $(x+y)^{\theta/\alpha}\leq 2^{\theta/\alpha}(x^{\theta/\alpha}+y^{\theta/\alpha})$
for $x,y\geq 0$, it is clear that 
$$
[u^+(V(\phi(k))]^{\theta/\alpha}\leq [C_1(V^+(\phi(k))^{\alpha}+1)]^{\theta/\alpha}\leq C_5 V^+(\phi(k))^{\theta}+C_5,
$$
with $C_5:=2^{\theta/\alpha}C_1^{\theta/\alpha}$ and, just like in \eqref{moon}, we get
$$
E[u^+(V(\phi(k)))]^{\theta/\alpha}\leq C_6+C_6(E_QV^-(\phi(k)))^{\theta}\leq C_6+C_6(E_Q|V(\phi(k))|)^{\theta},
$$
for $C_6:=C_5 [E(dP/dQ)^{\theta/(1-\theta)}]^{1-\theta}$. 
Since the right-hand side was shown to be bounded in $k$ (and convex
combinations do not change this), the family $u^+(V(\phi(k)))$, $k\in\mathbb{N}$ is unformly
integrable and $\sup_{\phi\in\mathcal{A}'(u)}Eu(V(\phi))<\infty$. 
Combining uniform integrability for $u^+(V(\phi(k)))$ with the Fatou lemma for $u^-(V(\phi(k)))$, we obtain
$$
Eu(V(\phi^*))\geq \lim_{k\to\infty} Eu(V(\phi(k))),
$$
which shows $\phi^*\in\mathcal{A}'(u)$ as well as \eqref{aki}.

Now define $\kappa_n:=p_n/(p_n+1)-\epsilon$. 
Applying Corollary \ref{mobutu} with the function
$$
u_n(x):=\kappa_n x+1,\ x<0,\quad u_1(x):=(1+x)^{\kappa_n},\ x\geq 0
$$
provides $Q\in\mathcal{M}$ with $dQ/dP\in L^{\infty}$ and $dP/dQ\in L^{\pi_n}$
with $\pi_n:=2/(1-\kappa_n)$. As $\pi_n=\pi_n(\epsilon)\to 2p_n+2$ when $\epsilon\to 0$,
for $p_{n+1}:=2p_n+2$ and for any $\epsilon>0$ we can assert the existence of $Q_{n+1}=Q_{n+1}(\epsilon)\in\mathcal{M}$ 
with $dQ_{n+1}/dP\in L^{\infty}$ and $dP/dQ_{n+1}\in L^{p_{n+1}-\epsilon}$. 

Iterating the above procedure we
get the existence of an optimiser for $\alpha<p_n/(p_n+1)$, for all $n$. This shows the
statement of the theorem because  $p_n/(p_n+1)\to 1$, $n\to\infty$.
\end{proof}

\begin{remark}{\rm When $u$ is strictly concave, a standard argument shows that $\phi^*$ is unique.}
\end{remark}

\begin{corollary}\label{horhos}
  Let Assumption \ref{relevant} be in force and assume 
$\sum_{i=1}^{\infty}b_i^2<\infty$. For each $p\geq 1$ there exists $Q=Q(p)\in\mathcal{M}$
such that $dQ/dP\in L^{\infty}$ and $dP/dQ\in L^p$.
Furthermore, $dQ/dP$ can be chosen of the form
\begin{equation}\label{monarch}
\frac{dQ}{dP}=\frac{u'\left(\sum_{i=1}^{\infty} \phi^*_i (\varepsilon_i-b_i)\right)}
{Eu'\left(\sum_{i=1}^{\infty} \phi^*_i (\varepsilon_i-b_i)\right)}
\end{equation}
with $u:\mathbb{R}\to\mathbb{R}$ strictly 
increasing, concave and continuously differentiable with bounded $u'$.
\end{corollary} 
\begin{proof}
This was shown during the proof of Theorem \ref{bumbo} since $p_n\to\infty$, $n\to\infty$. 
\end{proof}

\begin{remark} {\rm
In multiperiod discrete-time models with finitely many assets, if there is a risk-neutral
measure $Q\sim P$ then it can always be chosen to satisfy $dQ/dP\in L^{\infty}$. In continuous-time
models, however, finding such a $Q$ with bounded $P$-density is rather a rarity.

Also, in general, it is not easy to control the size of $dP/dQ$ either, see \cite{walter-rokhlin,rokhlin1}.
Corollary \ref{horhos} is a powerful result because it provides, in a model with countably many assets,
$Q\in\mathcal{M}$ with $dQ/dP,dP/dQ$ both satisfying strong conditions.
}
\end{remark}

\begin{remark} {\rm The standard route for attacking \eqref{force} is via its dual problem
where a convex functional is minimized over $\mathcal{M}$, see e.g. \cite{ks_exp,w}, and
the dual minimizer is given by formulas analogous to \eqref{monarch}, under appropriate conditions. 
Our approach works directly on the primal problem so
we do not need to introduce the dual setting. It has to be admitted though that our proof of
Theorem \ref{bumbo} heavily exploits duality in an indirect way: relative compactness of
a maximizer sequence is proved using some $Q\in\mathcal{M}$ with suitable integrability
properties.}
\end{remark}

\section{Almost optimal strategies}\label{no}

Theorem \ref{bumbo} establishes that there is $\phi^*\in\mathcal{A}'(u)$ attaining 
$\sup_{\phi\in\mathcal{A}'(u)}Eu(V(\phi))$.
In practical situations, however, only investments into finite market segments are feasible. It is hence an intriguing
question whether the value of the problem formulated over $\mathcal{E}$ equals that of the problem over
the whole domain, i.e. whether  
\begin{equation}\label{odo}
\sup_{\phi\in\mathcal{A}'(u)}Eu(V(\phi))=\sup_{\phi\in\mathcal{E}'(u)}Eu(V(\phi)),
\end{equation}
where $\mathcal{E}'(u):=\{\phi\in\mathcal{E}:\, Eu^-(V(\phi))<\infty\}$.

This problem seems hard, it is somewhat similar to the general question of whether, in a continuous
semimartingale model, the value of the utility maximization problem is attainable along
a sequence of ``simple'' strategies, where being simple may depend on the context. 

We report a partial result below. Throughout this section, $u:\mathbb{R}\to\mathbb{R}$ is a concave, nondecreasing function with $u(0)=0$.
Define $\Phi(x):=-u(-x)$, $x\geq 0$.

\begin{assumption}\label{moderate} Let 
\begin{equation}\label{ququ}
\lim_{x\to \infty} \Phi(x)/x=\infty,
\end{equation} 
\begin{equation}\label{huhu}
\limsup_{x\to\infty} \Phi(2x)/\Phi(x)<\infty
\end{equation}
hold. 
Define the conjugate function $\Psi(y):=\sup_{x\geq 0}\{xy-\Phi(x)\}$ and assume that
\begin{equation}\label{pupu}
\limsup_{y\to \infty}\Psi(2y)/\Psi(y)<\infty.
\end{equation}
\end{assumption}

\begin{remark} {\rm The above assumption means that $\Phi$, $\Psi$ are \emph{Young functions} (for $\Phi$ this
is stipulated by \eqref{ququ} and then it follows for $\Psi$ automatically)  
and both $\Phi$, $\Psi$ are \emph{of class $\Delta_2$} (this is the content of 
\eqref{huhu} and \eqref{pupu}), see \cite{neveu} and \cite{rrr} for definitions and details. We say that $u$ is \emph{moderate} if it satisfies
\eqref{ququ} and \eqref{huhu}.

A moderate utility function $u$ behaves in a ``power-like'' way near $-\infty$. Exponential $u$ (i.e. $u(x)=-e^{-x}$) is
a typical example of a utility function that is not moderate. Being moderate is, admittedly, a restrictive 
condition on $u$ but it still allows a large class of interesting cases. 
Condition \eqref{pupu} on $\Psi$ is rather mild, it is implied by the
standard ``reasonable elasticity condition'', see Corollary 4.2 of \cite{w}.}
\end{remark}

\begin{theorem}\label{koves} 
Let $u:\mathbb{R}\to\mathbb{R}$ be concave and nondecreasing with $u(0)=0$ such that Assumption \ref{moderate} is in force.
Then \eqref{odo} holds. Under Assumption \ref{relevant}, \eqref{manyineni} and \eqref{retes}, an optimal strategy $\phi^*$ exists
and 
$$
Eu(V(\phi^*))=\lim_{n\to\infty}Eu(V(\bar{\phi}(n))),
$$
where
\[
\bar{\phi}_j(n)=\phi^*_j,\ 1\leq j\leq n,\quad \bar{\phi}_j(n)=0,\ j>n,
\]
for each $n\geq 1$.
\end{theorem}
\begin{proof}
Fix $\phi\in\mathcal{A}'(u)$. The process
$$
Y_t:=-\sum_{j=1}^{\infty}\phi_jb_j+ \sum_{j=1}^t \phi_j\varepsilon_j,\ t\in\mathbb{N}\cup\{\infty\}
$$
is a (convergent) martingale with respect to the filtration 
$$
\mathcal{F}_0:=\{\emptyset,\Omega\},\ \mathcal{F}_t:=\sigma(\varepsilon_1,\ldots,\varepsilon_t),\ \mathcal{F}_{\infty}:=\sigma(\varepsilon_j,\, j\in\mathbb{N}),
$$
so $Z_t:=Y_t^-$ is a submartingale.
(Note that the parameter $t$ here has no interpretation as ``time''.) 

Defining $\Phi(x):=-u(-x)$, $x\geq 0$, $\phi\in\mathcal{A}'(u)$ entails
$E\Phi(Z_{\infty})<\infty$. 
Since \eqref{ququ} and \eqref{pupu} hold, this implies 
\begin{equation*}
E\Phi(\delta\sup_n Z_n)<\infty
\end{equation*}
for some $\delta>0$, see Proposition A-3-4 of \cite{neveu} and the ensuing discussion. (In fact, that result
is stated for martingales there but the proof works for non-negative submartingales in the same manner.)

Using \eqref{huhu}, we easily deduce
\begin{equation}\label{tutt}
E\Phi(\Delta\sup_n Z_n)<\infty
\end{equation}
for every $\Delta>0$ as well. Define
\[
{\phi}_j(n)=\phi_j,\ 1\leq j\leq n,\quad {\phi}_j(n)=0,\ j>n.
\]
Note that $V^-(\phi(n))\leq Z_n+\sum_{j=1}^{\infty}\vert \phi_j b_j\vert\leq Z_n + Q$ where
$Q:=\Vert\phi\Vert_{\ell_2}\Vert b\Vert_{\ell_2}$.
From \eqref{tutt} we infer that
\begin{equation}\label{tu}
E\Phi(\sup_n V^-(\phi(n)))\leq \frac{1}{2}\left(E\Phi(2\sup_n Z_n)+\Phi(2Q)\right)<\infty,
\end{equation}
by convexity of $\Phi$. In particular, $\phi(n)\in\mathcal{E}'(u)$, $n\in\mathbb{N}$.

Dominated convergence implies $Eu^-(V(\phi(n)))\to Eu^-(V(\phi))$
by \eqref{tu},
while Fatou's lemma implies $Eu^+(V(\phi))\leq \liminf_{n\to\infty} Eu^+(V(\phi(n)))$ so
\[
\sup_{\phi\in\mathcal{A}'(u)}Eu(V(\phi))\leq\sup_{\phi\in\mathcal{E}'(u)}Eu(V(\phi))
\]
follows. The other inequality being trivial, we arrive at \eqref{odo}.

Under Assumption \ref{relevant}, \eqref{manyineni} and \eqref{retes} there is
an optimal $\phi^*$ by Theorem \ref{bumbo}. Applying the above argument to $\phi^*=\phi$
we get the second statement of the theorem, too.
\end{proof}

\noindent\textbf{Acknowledgments.} Support from the
``Lend\"ulet'' Grant LP2015-6 of the Hungarian Academy of Sciences is gratefully acknowledged.
The idea of this paper was conceived during a research visit in 2015 at Dublin City University; I 
thank Paolo Guasoni for his invitation. I am also grateful Josef Teichmann for his kind invitation
to ETH, Z\"urich in 2014: it was discussions with him that renewed my interest in the models
treated here. Finally, I sincerely thank the anonymous referees for their very helpful comments.

\end{document}